\documentclass[letterpaper, 10 pt, conference]{ieeeconf}
\usepackage{graphicx}

\IEEEoverridecommandlockouts                              
\usepackage{amsmath} 
\usepackage{amssymb}  
\usepackage{amsfonts}
\usepackage{stmaryrd}
\usepackage{algorithmic,algorithm}
\newtheorem{theorem}{Theorem}

\newtheorem{proposition}{Proposition}

\title{\LARGE \bf Dynamic multi-agent assignment via discrete optimal transport}

\author{Koray Gabriel Kachar* and Alex Gorodetsky**
\thanks{*Koray G. Kachar is with the Department of Aerospace Engineering, University of Michigan, Ann Arbor, MI 48109, USA {\tt\small kkachar@umich.edu}}%
\thanks{**Alex Gorodetsky is an Assitant Professor, Aerospace Engineering, University of Michigan, Ann Arbor, MI 48109, USA {\tt\small goroda@umich.edu}}
}

\newcommand{\reals}{\mathbb{R}}
\newcommand{\posint}{\mathbb{Z}_{+}}
\renewcommand{\vec}[1]{#1}
\newcommand{\mat}[1]{#1}

\newcommand{\ta}{\tilde{\mat{A}}}

\DeclareMathOperator*{\argmin}{arg\,min}

\usepackage{todonotes}

\begin{document}

\maketitle
\thispagestyle{empty}
\pagestyle{empty}

\begin{abstract}

We propose an optimal solution to a deterministic dynamic assignment problem by leveraging connections to the theory of discrete optimal transport to convert the combinatorial assignment problem into a tractable linear program. We seek to allow a multi-vehicle swarm to accomplish a dynamically changing task, for example tracking a multi-target swarm. Our approach simultaneously determines the optimal assignment and the control of the individual agents. As a result, the assignment policy accounts for the dynamics and capabilities of a heterogeneous set of agents and targets. In contrast to a majority of existing assignment schemes, this approach improves upon distance-based metrics for assignments by considering cost metrics that account for the underlying dynamics manifold. We provide a theoretical justification for the reformulation of this problem, and show that the minimizer of the dynamic assignment problem is equivalent to the minimizer of the associated Monge problem arising in optimal transport. We prove that by accounting for dynamics, we only require computing an assignment once over the operating lifetime --- significantly decreasing computational expense. Furthermore, we show that the cost benefits achieved by our approach increase as the swarm size increases, achieving almost 50\% cost reduction compared with distance-based metrics. We demonstrate our approach through simulation on several linear and linearized problems. 

\end{abstract}

\section{INTRODUCTION}

Our aim is to enable efficient centralized decision making amongst swarms of agents that are tasked to intercept or track a swarm of target vehicles. Specifically, we seek an optimal centralized assignment policy that is \textit{capability-aware} --- it can \textit{leverage} known dynamics of the agents and targets to make optimal assignments that respect the capabilities of the agents and targets. We approach this problem by posing an objective function that accounts for both the high level cost of all assignments and the low-level costs of the optimal control policies used by each agent measures. We add differential constraints arising from vehicle dynamics to complete the optimization formulation. This approach stands in contrast to the majority of techniques that use \textit{distance-based} (or bottle-neck assignment~\cite{Burkard2002}) cost functions~\cite{Kloder2006,Faigl2012,Panagou2014}.


The approach we take in this work is based on the realization of the close relationship between the given problem and the theory of optimal couplings, or optimal transport~\cite{Villani2003,Villani2008}. In the context of probability theory, to which it is often applied, optimal transport studies the problem of determining joint distributions between sets of random variables whose marginal distributions are constrained. In other words it tries to find a coupling that maps a reference measure to a target measure. Optimal transport has been applied a wide variety of other areas as well; for instance it has been used to great effect in the areas of machine learning~\cite{Canas2012,Cuturi2013,Frogner2015}, image manipulation~\cite{Ferradans2014}, and Bayesian inference~\cite{Moselhy2012}.

\subsection{Innovation and Contributions}
The fundamental insight we use to relate OT to the present context is that the set agents may be viewed as a discrete measure that we seek to map to the discrete measure denoted by the set of targets. In this way, we consider discrete optimal transport (DOT). Our context is also different from the standard DOT problem in that the target measure is changing and that the transport of the reference to the target must respect the differential constraints given by the dynamics. Our innovation is that we can address these issues by introducing a new metric that respects the dynamics, as explored by Ghoussoub et. al. \cite{ghoussoub2018optimal}, rather than the traditional unweighted Euclidean metric that underpins the Wasserstein or ``Earth Movers'' distance. Our proposed metric uses the optimal control cost of a single-agent vs. single-target system as the cost of the proposed assignment. For instance if the agents will perform LQR reference tracking to intercept their targets then the LQR cost is used as the transportation cost. Alternatively, if the agents will solve a  pursuit evasion game, then the transportation cost will be obtained from the solution to the differential game. In this way, the assignment becomes aware of the capabilities of the system, including the differential constraints and the decision making approach of individual agents.

Our problem is specified by two inputs
\begin{enumerate}
\item The dynamics of the agents and their targets
\item A mechanism to evaluate a feedback policy and its cost for any single agent
\end{enumerate}
Using these two specifications we form a cost function that is the sum of all individual agent cost functions, and seek an assignment that minimizes this total cost. Critically, we see that the cost used for each agent is that of the \textit{feedback policy} --- not the distance. Typically, such feedback policies are obtained to either optimally regulate or operate the underlying agent. Thus, the cost incurred by an agent that is following its feedback policy is a more appropriate measure of optimality than one based on the distance an agent must travel.

Our approach provides a solution for this problem, and consists of the following contributions
\begin{enumerate}
    \item A new capability-aware assignment approach that simultaneously optimizes the assignment and the underlying feedback controls of the agents
    \item A reformulation of the vehicle-target dynamic assignment problem as a linear program by leveraging concepts from discrete optimal transport
\end{enumerate}
The above two contributions are supported via both theoretical and simulation results. In particular, we prove that our cost function can be reformulated into the Monge problem from optimal transport. This problem can then be solved via a linear programming approach. The capability-aware assignment problem is demonstrated to have lower final cost as compared to a distance-based assignment that neglects the feedback control policy. We empirically show that the optimality gap between our approach and distance-based metrics grows with the number of agents. Finally, we prove that after formulating the assignment problem in the DOT framework, it needs only be solved once rather than repeatedly over the life of the system. As a result, we see significant computational benefits compared to repeatedly re-calculating a distance-based assignment.

\subsection{Related work}


Assignment and resource allocation problems present themselves across many disciplines. In the area of multirobot task assignment, self-organizing map  neural-networks were designed to learn how to assign teams of robots to target locations under a dual decision-making and path-planning framework \cite{zhu2006neural}. However, the algorithm proposed in that work is largely heuristic and does not consider the underlying capabilities of the assigned robots.  Other papers have considered more general kinematic movements of the formations in general, rather than individual agent capabilities, and were able to provide suboptimality guarantees for the overall assignment~\cite{ji2006role}. Another approach can be found on~\cite{kwok2002analyzing} that proposes an approach that is very similar to ours in that it solves a related linear programming problem. However, that approach did not consider the effect of general dynamics of the system or a changing set of targets. 

A similar assignment problem arises in from vehicle-passenger transit scheduling that have become extremely important in ride sharing applications \cite{ceder2011optimal},  \cite{de2016solving}. Alonso-Mora et.al \cite{alonso2017demand} investigated dynamic trip-vehicle assignment that is both scalable and optimal using a greedy guess that refines itself over time. In general, these problems lack consideration of underlying dynamics of the resource  being assigned or that task being assigned to. Assignment problems also arise in wide areas of econometrics dealing with matching population models to maximize total utility surplus, contract theory, or risk management \cite{galichon2018optimal},  \cite{graham2011econometric},  \cite{jonker1987shortest}. In general, these problems also do not consider the underlying dynamic nature within the assignment problem.
 
One closely related application area that, at times, also considers the dynamics in completing an assignment is the so-called weapon-target-assignment (WTA) ~\cite{Hosein1990,Cai2006}.
The WTA problem itself comes in two-forms: the static WTA and the dynamic WTA. In the static WTA problem, all of the assignments are made at once with no feedback possible, whereas the dynamic WTA allows for feedback between decision making phases~\cite{xin2010efficient}. Our approach is related to this problem as it is a certain mixture of these two; first, it considers explicit dynamic capabilities of the agents and targets during the assignment problem; and second it potentially allows for reassignment of the agents during operations. Our setup can also be viewed as a limiting case of the traditional WTA in that we assume that once the weapon intercepts the target it successfully destroys it with 100\% probability. This contrasts to the traditional WTA setting where a weapon might only have a certain probability of destroying its target.
 
The traditional WTA assignment problem (with probabilistic survival of targets after interception) has typically been formulated as a nonlinear integer programming problem for which exact methods have not yet been found. As a result, a large number of heuristic and/or approximate approaches have been developed. For instance approaches based on approximate dynamic programming~\cite{davis2017approximate}, fuzzy reasoning~\cite{csahin2014approximating},  various search heuristics~\cite{xin2010efficient}, genetic and/or global optimization approaches~\cite{lee2002immunity, lee2003efficiently}, network-based heuristics~\cite{Ahuja2007}, amongst others have all been studied. 
In comparison to these previous works on WTA we provide several contributions. Our proposed (as far as we are aware previously unrecognized) link to optimal transport theory can yields additional theoretical and computational guarantees.

Finally, we review some connections between our proposed approach and existing solutions in robotics and control. Fredrick et. al. \cite{frederick2019multi} investigated multi-robot path planning and shape formation underpinned by optimal transport to prove that the desired formations can be obtained while maintaining collision-free motions with assurance of convergence to global minima. Similarly, Bandyopadhyay et. al. \cite{bandyopadhyay2014probabilistic, bandyopadhyay2016novel, bandyopadhyay2017probabilistic} describe an approach where swarm members are placed into bins which have constraints that must satisfied in order to permit the transition of the agents to neighboring bins. These motion constraints are representative of the dynamics or physical limitations present in the system. In terms of the approach described in this paper, the optimal transport cost metric is thus a modified $L_2$ distance between the centroids of the bins subject to a motion constraint matrix; if motion is possible, the cost is the $L_2$ distance, otherwise the cost is the maximum value.

Here we consider a specific setting, one with deterministic and known dynamics, for which we can prove optimality. While we do not consider limitations on communication between agents, this problem has also been consider in the decentralized decision making context where each vehicle is making its own decisions. In this case all of the agents must come to a consensus through different communication/negotiation strategies, see e.g.~\cite{Yamauchi1999} for a greedy assignment strategy and~\cite{arslan2007autonomous} for an example of a game-theoretical formulation.

\section{Problem Definition}\label{sec:problem}

In this section we define the dynamic agent-target assignment problem. We begin by describing a dynamical system that describes the evolution of active states, targets, and destinations. We then provide an optimization problem that we seek to solve.

\subsection{Dynamical System}\label{sec:dyn}

We limit our presentation to the case of linear, control-affine systems for clarity of exposition. Our approach and theory is also valid for nonlinear systems given an ability to compute policy costs associated with nonlinear controllers.

Let $n,m \in \posint$ denote a positive number of autonomous agents (resources) and targets.  If we consider agent $i$ and target $j$, then their states at time $t$ are denoted by $x_{i}(t) \in \mathcal{X} \subset \reals^{d}$, $y_{j}(t) \in \mathcal{Y} \subset \reals^{d_y}$, respectively.  Agent $i$ takes actions $u_i \in \reals^{d_{u}}.$ 

In our problem, the number of agents and targets can only decrease with time. We leave consideration of newly appearing targets to future work. Each object that has not been removed is termed \textit{active} so that at $t=0$ we have $n$ active agents and $m$ active targets.

An agent/target pair can become inactive when the agent successfully intercepts or completes its resource allocation. Let $\pi_x$ and $\pi_y$ define functions that extract the \textit{positions} of the agents and  targets from their states. Successful resource allocation of an agent is defined when $\lVert \pi_x(x_i) - \pi_y(y_j) \rVert \leq \epsilon.$ In other words, when the position of agent $i$ is within an $\epsilon-$ball of target $j$, then both become inactive.

The activity of the agents and targets at each time is represented by the active sets $\mathcal{A}(t)$ and $\mathcal{T}(t).$ For instance, if all agents are active then $\mathcal{A}(t) = [1, \ldots, n]$, whereas if, for instance agent $i$ has successfully reached target $j$ then $i$ is removed from $\mathcal{A}$ and $j$ is removed from $\mathcal{T}$. This process defines an evolution of the active sets. At a given time $t$, the active agents and targets evolve according to a stochastic differential equation
\begin{align*}
    \dot{x}_{i}(t) &= \mat{A}_i x_i(t) + \mat{B}_i u_i(t) 
    &\textrm{for } i \in \mathcal{A}(t) \\ 
    \dot{y}_{j}(t) &= \ta_{i} y_j(t)  
    &\textrm{for } j \in \mathcal{T}(t).
\end{align*}

where $\mat{A}_i \in \reals^{d \times d}$ and $\mat{B}_i \in \reals^{d \times d_u}$ correspond to the drift of linear dynamics of the agent; and $\ta_j \in \reals^{d_y \times d_{yu}}$ corresponds to the closed loop linear dynamics of the targets. Note that here we have assumed linear dynamics; however this assumption is entirely unnecessary for our theory in Section~\ref{sec:analysis}. It is however, more computationally tractable because it leads to a solution of sets of linear optimal control problems, for instance LQR or LQI. A more significant assumption that these dynamics imply is that there is no interaction between agents, i.e., we do not consider collisions or other interference effects. We leave this matter for future work, but note that in the simulation examples in Section~\ref{sec:experiments} we noticed that collisions between agents did not generally occur.

Finally, the entire state of the system be defined by the tuple $S(t) = \left(\mathcal{A}(t), \mathcal{T}(t), \left(x_i(t)\right)_{i_\in \mathcal{A}(t)}, \left(y_j(t)\right)_{j \in \mathcal{T}(t)}\right).$ Let $\mathcal{O}$ denote a set of states $S(t)$ for which there is at least one active target or destination. Define the \textit{exit time} $\tau$ as the first time that the state $S(t)$ exists from $\mathcal{O}$.

\subsection{Policies, cost functions, and optimization}


We seek a feedback policy $\mu$ that maps $S(t)$ to a set of controls $u_i(t)$ for all active agents $i \in \mathcal{A}(t)$. 
The policy is represented by a tuple 
\begin{equation}\label{eq:pol}
    \mu(S(t)) = \left(\sigma_t, ( \mu_i : i \in \mathcal{A}(t)) \right)
\end{equation}
where $\sigma_t: \mathcal{A}(t) \to \mathcal{T}(t)$ is an index function that assigns active agents to active targets and $\mu_i$ is a feedback control policy for the individual agents. The goal then, is to determine an optimal feedback policy of this form.

An \textit{optimal} feedback policy is one that minimizes
\begin{equation} \label{eq:opt}
    \displaystyle{\mu^* = \argmin_{\mu} \left[ 
     \int_{0}^\tau g(t, S(t), \mu(S(t)))\right] 
    },
\end{equation}
where the stage cost mapping the state of the system to the reals is $g$ and the time $\tau$ is the first exit time. The optimal value function will then be denoted by 
\begin{equation}\label{eq:opt_val}
J(S(0)) = \int_{0}^{\tau} g(t, S(t), \mu^{*}(S(t)) dt
\end{equation}



The stage cost intends to guide each agent $i$ to its assigned target $j$ and is therefore represented by the sum
\begin{equation}\label{eq:stage-sep}
    g(t, S(t), \mu(S(t))) = \sum^{n}_{i=1, i \in \mathcal{A}(t)} g_i(x_i, u_i, y_{\sigma_t(i)}),
\end{equation}
where the cost $g_i$ assigned to agent $i$ is a function of the corresponding agent state, the agent control, and the target to which the agent is assigned. For instance, this cost could be a quadratic corresponding to an infinite horizon tracking problem \cite{willems2004rigorous}
\begin{equation}\label{eq:stage}
\begin{split}
   & g_i(x_i, u_i, y_{\sigma_t(i)}) = \\
    & \left(\pi_x(x_i) - \pi_y(y_{\sigma_t(i)})\right)^T Q  \left(\pi_x(x_i) - \pi_y(y_{\sigma_t(i)})\right) + 
    u_i^T R u_i\\
    & -\left(\pi_{x_{ss}}(x_i) - \pi_{y_{ss}}(y_{\sigma_t(i)})\right)^T Q  \left(\pi_{x_{ss}}(x_i) - \pi_{y_{ss}}(y_{\sigma_t(i)})\right)\\
    & + 
    u_{{ss}_i}^T R u_{{ss}_i}
\end{split} 
\end{equation}
where $Q$ and $R$ penalize the distance between the agent-target system and the control and $\pi_{x_{ss}}$ and $\pi_{y_{ss}}$ are the steady-state values for the agent and assigned-to target, respectively. These transient and steady-state terms represent the dual goals of this particular optimal controller, which are to drive the error of the agent-target to the optimal steady-state and then keep the agent-target system at this optimal state.

\section{Discrete optimal transport}\label{sec:dot}

In this section we provide background to discrete optimal transport, and indicate how it relates to our dynamic assignment problem. We follow the description given by~\cite{Peyre2019}.

Let $\Sigma_{n}$ denote a probability simplex so that any  $\vec{a} \in \Sigma_{n}$ belongs to the simplex
\begin{equation}
    \Sigma_{n} \equiv \left \{ a \in \reals_{+}^n : \sum_{i=1}^n a_i=1  \right\}
\end{equation}
A discrete measure $\mu_X$ is defined by the fixed locations $x_1,\ldots,x_n \in \mathcal{X}$ and weights $\vec{a} \in \Sigma_{n}$, and denoted by
\begin{equation}
    \mu_X = \sum_{i=1}^n a_i \delta_{x_i}.
\end{equation}

The transport map $T$ between two discrete measures, $\mu_X = \sum_{i=1}^n a_i \delta_{x_i}$ and $\mu_Y = \sum_{j=1}^m b_j \delta_{y_j},$ is a surjective function $T:\{x_1,\ldots,x_{n}\} \to \{y_1,\ldots,y_{m}\}$ that satisfies
\begin{equation}
    b_j = \sum_{i,T(x_i) = y_j} a_i  \quad \forall j = 1,\ldots, m.
\end{equation}
Compactly the above is written as $T_{\sharp}\mu_x = \mu_y$, which implies that measure $\mu_y$ is the \textit{push-forward} measure of $\mu_x$ under the map $T$.

\subsection{Monge problem}

We seek to find \textit{optimal} assignments for the agent-target system, and this implies that we seek a map $T$ that minimizes some the transportation cost. Let the transportation cost be defined pointwise as
\begin{equation}\label{eq:tcost}
c(x,y) = c(x, T(x)), \quad \forall (x, y) \in \mathcal{X} \times \mathcal{Y}.
\end{equation}
The Monge problem then seeks a map $T$ that minimizes 
\begin{equation}\label{eq:monge}
\min_{T}\left[ \sum_{i=1}^n c(x_i, T(x_i)) : T_{\sharp}\mu_X = \mu_Y \right].
\end{equation}
To parameterize T, we can define an index function $\sigma:\llbracket n \rrbracket \to \llbracket m \rrbracket$ so that $T(x_i) = y_{\sigma(i)}$, just as in Equation~\eqref{eq:pol}.

The problem with optimizing Equation~\eqref{eq:monge} is that it is not-convex. In general, convexity can be achieved by relaxing the deterministic nature of the map to allow portions of $x_i$ to be directed towards $y_j$. The resulting stochastic map is defined by a coupling matrix (transition probability matrix, or stochastic matrix) $\mat{P} \in \reals_{+}^{n \times m}$ with $p_{ij}$ indicating the portion of $x_i$ assigned to $y_j$. Define a set of allowable coupling matrices
\[
\mathcal{U}(\vec{a}, \vec{b}) \equiv \left\{ \mat{P} \in \reals_{+}^{n \times m}: \mat{P} \vec{1}_m = \vec{a} \textrm{ and } \mat{P}^T \vec{1}_n = b \right\},
\]
where $\vec{1}_{k}$ denotes a vector of size $k$ filled with ones.
The Monge-Kantorovich optimization formulation then becomes
\begin{equation}\label{eq:kantorovich}
    L_c(\vec{a},\vec{b}) = \min_{\mat{P} \in \mathcal{U}(\vec{a},\vec{b})} \sum_{i,j} c(x_i, y_j)p_{ij},
\end{equation}
and it can be solved with linear programming. Under the conditions given next, the solution to this optimization problem is equal to the solution of the Monge problem.

\subsection{Matching problem}
The matching problem is a particular realization of OT that has the property that the minimizer of~\eqref{eq:kantorovich} is equal to that of~\eqref{eq:monge}. The formal statement of this equivalence is given below.

\begin{proposition}[Kantorovich for matching (Prop 2.1,~\cite{Peyre2019})]
If $n=m$ and $\vec{a} = \vec{b} = \vec{1}_n/n$ then there exists an optimal solution for minimizing Equation~\eqref{eq:kantorovich} $\mat{P}_{\sigma^*}$, which is a permutation matrix associated to an optimal permutation $\sigma^{*} \in \textrm{Perm(n)}$ for Equation~\eqref{eq:monge}.
\end{proposition}

In this setting, we seek a a one-to-one mapping $T$. The constraint set   $\mathcal{U}(\vec{a}, \vec{b})$ becomes the set of doubly stochastic matrices with entries, and the coupling matrix has elements
\begin{equation}\label{eq:coupling}
    p_{ij} = \left\{ 
    \begin{array}{cc}
    1/n & \textrm{ if } j = \sigma(i), \\
    0 & \textrm{ otherwise}.
    \end{array}
    \right.
\end{equation}
In the context of our assignment problem, this case occurs when there are an equal amount of agents and targets. A discrete optimal transport formulation can also be applied for a relaxation of the Kantorovich problem so that several agents can be assigned to the same target. This problem can also guarantee binary coupling matrices (essential for our application). For further details, we refer to~\cite{Ferradans2014}.

\subsection{Metrics}
The choice of cost $c(x_i, y_j)$ is problem dependent; however, the most commonly used cost for optimal transport between distributions is the Euclidean distance. Parameterized, by $\alpha \geq 1$, it is given by
\begin{equation}\label{eq:euclid}
    c(x_i, y_j) = \lVert x_i - y_j \rVert^{\alpha},
\end{equation}
where $\lVert \cdot \rVert$ is the Euclidean norm. Using this metric for points, 
$L_c(\vec{a},\vec{b})$ can be viewed as a metric for measures. In this case, $L_c$ is called the $L^\alpha$-Wasserstein distance~\cite{Ferradans2014}. In the statistical community, this metric is also called the \textit{Earth movers distance} (EMD).

This metric implies that the cost of moving a resource $x_i$ to $y_j$ is dominated by the distance between them in Euclidean space. The total cost of the assignment then becomes a sum of the distances. In our application to assignment in dynamical systems, the Euclidean metric may not be most appropriate because it does not account for the underlying dynamics of the system. One of our insights is that using a metric determined by the underlying dynamics of the problems leads to more optimal assignments. 


\section{Assignment in dynamic systems with DOT}\label{sec:dyndot}

In this section we describe how DOT can be applied to minimize Equation~\eqref{eq:opt}. As previously stated, our goals are to determine the assignment policy $\sigma$ that is ``capability-aware.'' In other words, the assignment policy must account for the dynamics of the system --- the capabilities of the agents and the targets. 

A direct application of the EMD metric within DOT would potentially require constant reassignment at each timestep because the metric makes no accountability of the future system state. In other words, it would be \textit{greedy} and simply assign each agent to minimize the total distance
between agent/target pairs.

In the next two subsections we first describe an algorithm that leverages the knowledge of the interception strategies of each agent to make an assignment, and then we provide and discuss pseudocode to illustrate the flexibility of our approach .







\subsection{Algorithm}\label{sec:mancost}

The metric we propose for the transportation cost of assigning $x_i$ to $y_j$ is that corresponding to the optimal actions of a one-agent-to-one-target optimization problem. For instance, let us assume that agent $i$ is paired to target $j = \sigma(i)$, then in the 1v1 scenario for policy $\mu_i$, we have a total incurred cost of 
\begin{equation}\label{eq:1v1cost}
\begin{split}
    J^{\mu_i}_{i\sigma(i)}(x_i(0), &y_{\sigma(i)}(0)) = \\  
    \int_{0}^{\tau_{i}} & g_{i}(t, x_i, \mu_i(t, x_i(t), y_{\sigma(i)}(t)), y_{\sigma(i)}) dt 
\end{split} 
\end{equation}
The optimal policy is obtained by minimizing this value function.

Let $J_{i\sigma(i)} = \min_{\mu} J_{i\sigma(i)}^{\mu_i}$ correspond to the value function under the optimal policy. Our proposed transportation cost is this optimal value function
\begin{equation}\label{eq:cdyn}
    c_{\textrm{dyn}}(x_i(0), y_{\sigma(i)}(0)) = J_{i\sigma(i)}(x_i(0), y_{\sigma(i)}(0)).
\end{equation}

For example, for linear dynamics with quadratic cost the transportation becomes 
\begin{align}
c_{\textrm{dyn}}&(x_i(0),  y_{\sigma(i)}(0)) =  \nonumber \\  &f(x_{i}(0),y_{\sigma(i)}(0))^T P_{i\sigma(i)} f(x_{i}(0),y_{\sigma(i)}(0)) + \nonumber \\
&\quad  2p_{i\sigma(i)}f(x_{i}(0), y_{\sigma(i)}(0)) - \nonumber \\
&\qquad f(x_{i}(0),y_{\sigma(i)})_{ss}^T P_{i\sigma(i)} f(x_{i}(0),y_{\sigma(i)})_{ss} - \nonumber \\
&\quad \qquad 2p_{i\sigma(i)}f(x_{i}(0),y_{\sigma(i)})_{ss} \label{eq:cdyn_lqt}
\end{align}
where $f$ is a function that combines the agent and target state into suitable form. For instance $f(x,y) = x - y$ can be used for a reference tracking problem \cite{willems2004rigorous}; $P_{i\sigma(i)}$ is the solution of the continuous algebraic Riccati equation for the LQR-based tracker; $p_{i\sigma(i)}$ is the feed-forward control of the state being tracked; and  $f(x_{i}(0),y_{\sigma(i)})_{ss}$ is a function that provides the steady-state value for the quadratic agent-target state.

\subsection{Pseudocode}

In this section we provide and describe the pseudocode for the proposed algorithm. A sample implementation that makes specific choices about the dynamics and policies is shown in Algorithm~\ref{alg:dyn}. This algorithm takes as inputs all of the agent states, target states, and dynamics. In Line~\ref{alg:dyn:assign}, the assignment and individual agent policies are obtained by querying Algorithm~\ref{alg:assign}.

Algorithm~\ref{alg:assign} performs the optimal transport allocation. Its inputs are all of the states and their dynamics, an algorithm for computing the policies for each agent when it is assigned to some state, and a cost metric. Algorithm~\ref{alg:dyn} makes two specific choices for these components. First, it uses the linear quadratic tracker (LQT) developed in~\cite{willems2004rigorous} that uses linear dynamics. However, if the dynamics are nonlinear, any other computable policy can be used. The specific cost metric is $\texttt{dist}_{dyn}$, which is the dynamics-based distance given by Algorithm~\ref{alg:dyn_cost}. This algorithm, uses the cost of the LQT policy~\eqref{eq:cdyn_lqt} as the transportation cost. Algorithm~\ref{alg:assign} has two steps. First it calls the discrete optimal transport routine with a pre-specified distance metric to obtain an assignment $\sigma$. It then iterates through all agents and obtains the individual policy for each agent that follows the assignment.

The high-level Algorithms~\ref{alg:dyn} and~\ref{alg:emd} demonstrate the differences between our approach and the standard approach that uses the distance metric (Algorithm~\ref{alg:emd_cost}). Algorithm~\ref{alg:emd_cost} evaluates the distance directly by extracting the positions ($\texttt{pos}$, but an entire state can also be used), and discards $g_{pol}$ --- the cost of the actual policy.  As a result, this assignment needs to be continuous recomputed. In Section~\ref{sec:analysis} we prove that our approach only requires an assignment to be generated once. 

\renewcommand{\algorithmicrequire}{\textbf{Inputs}}
\renewcommand{\algorithmicensure}{\textbf{Outputs}}
\begin{algorithm}
 \caption{Simulation engine with dynamics-based assignment and LQT tracking}
\label{alg:dyn}  
\begin{algorithmic}
  \REQUIRE $n$ agents $\mathcal{A}$, states $X=(x_i)_{i \in \mathcal{A}}$, dynamics $(f_{i})_{i \in \mathcal{A}}$;\\
  \qquad \ $n$ targets $Y = \{y_1,\ldots, y_{n}\}$,  dynamics $(f_{T_i})_{i=1}^n$
  \ENSURE Completion of simulation, all targets are tracked
  \STATE {\footnotesize $\sigma,(\mu_{i})_{i=1}^n = \texttt{assign}(X,Y,(f_i)_{i=1}^n,(f_{T_{i}})_{i-1}^n, LQT, \texttt{dist}_{dyn})$ } \label{alg:dyn:assign}
  \WHILE{$t \leq t_f$}
  \FOR{$(i, x_i) \in \mathcal{A}$}
  \STATE $j = \sigma(i)$  
  \STATE $y_j(t) \leftarrow$ Get target state from environment
  \STATE $u_i(t) = \mu_i(t, x_i, y_j)$
  \STATE $x_i^*(t+dt) = ODEsolve(t, t+dt,f_i, x_i, u_i)$
  \STATE $\mathcal{A} = \texttt{remove-tracked}(\mathcal{A}, (x_{i})_{i=1}^n, (y_{i})_{j=1}^n)$
  \ENDFOR
\ENDWHILE
\end{algorithmic}
\end{algorithm}

\begin{algorithm}
\caption{Assignment policy: $\texttt{assign}$}
\label{alg:assign}
\begin{algorithmic}[1]
  \REQUIRE Set of agents $X = (x_i)_{i=1}^n$; Set of targets $Y = (y_i)_{i=1}^n$; agent dynamics $f_i$; target dynamics $f_{T_i}$; control policy generator $\texttt{pol}$; 1 vs. 1 function cost calculator $\texttt{dist}$  \\
  \ENSURE An assignment policy $\sigma$; agent policies $(\mu_{i})_{i=1}^n$
  \STATE $\sigma = \texttt{DOT}(X, Y, \texttt{dist})$
  \FOR{$i=1\ldots n$}
  \STATE $g_i, \mu_i = \texttt{dist}(x_i, y_{\sigma_{i}}, f_i,f_{T_{\sigma_{i}}})$
  \ENDFOR
  \RETURN $\sigma, (\mu_{i})_{i=1}^n$
\end{algorithmic}
\label{alg:assignment}
\end{algorithm}

\begin{algorithm}
\caption{Dynamics-based cost: $\texttt{dist}_{dyn}$}
\label{alg:dyn_cost}
\begin{algorithmic}[1]
  \REQUIRE Agent state $x$; target state $y$; agent dynamics $f$; target dynamics $f_T$; control policy generator $\texttt{pol}$  \\
  \ENSURE The cost $g$ and the optimal policy $\mu$ for assigning agent $x$ to target $y$
  \STATE $g, \mu = \texttt{pol}(x, y, f, f_T)$
  \RETURN $g, \mu$
\end{algorithmic}
\end{algorithm}

\begin{algorithm}
\caption{Distance-based cost: $\texttt{dist}_{emd}$}
\label{alg:emd_cost}
\begin{algorithmic}[1]
  \REQUIRE Agent state $x$; target state $y$; agent dynamics $f$; target dynamics $f_T$; control policy generator $\texttt{pol}$  \\
  \ENSURE The cost $g$ and the optimal policy $\mu$ for assigning agent $x$ to target $y$
  \STATE $g_{pol}, \mu = \texttt{pol}(x, y, f, f_T)$
  \STATE $g = \lVert \texttt{pos}(x) - \texttt{pos}(y) \rVert$
  \RETURN $g, \mu$
\end{algorithmic}
\end{algorithm}

\begin{algorithm}
 \caption{Simulation engine with distance-based assignment and LQT tracking}
\label{alg:emd}  
\begin{algorithmic}
  \REQUIRE $n$ agents $\mathcal{A}$, states $X=(x_i)_{i \in \mathcal{A}}$, dynamics $(f_{i})_{i \in \mathcal{A}}$;\\
  \qquad \ $n$ targets $Y = \{y_1,\ldots, y_{n}\}$,  dynamics $(f_{T_i})_{i=1}^n$
  \ENSURE Completion of simulation, all targets are tracked
  \WHILE{$t \leq t_f$}
  \STATE $X, Y \leftarrow \texttt{extract-active-agents}(\mathcal{A})$
  \STATE {\footnotesize $\sigma,(\mu_{i})_{i=1}^n = \texttt{assign}(X,Y,(f_i)_{i=1}^n,(f_{T_{i}})_{i-1}^n, LQT, \texttt{dist}_{emd})$} 
  \FOR{$(i, x_i) \in \mathcal{A}$}
  \STATE $j = \sigma(i)$  
  \STATE $y_j(t) \leftarrow$ Get target state from environment
  \STATE $u_i(t) = \mu_i(t, x_i, y_j)$
  \STATE $x_i^*(t+dt) = ODEsolve(t, t+dt,f_i, x_i, u_i)$
  \STATE $\mathcal{A} = \texttt{remove-tracked}(\mathcal{A}, (x_{i})_{i=1}^n, (y_{i})_{j=1}^n)$
  \ENDFOR
\ENDWHILE
\end{algorithmic}
\end{algorithm}

\section{Analysis}\label{sec:analysis}
In this section we analyze the proposed algorithms. Our aim is to show that the optimization problem~\eqref{eq:opt} can be reformulated into the Monge-Kantorovich optimal transport problem. We follow a two step procedure. First we show that the optimal assignment policy of Equation~\eqref{eq:pol} does not change with time, and next we show that the problem is identical to the Monge problem.

\begin{proposition}[Constant assignment policy]
The optimal assignment policy $\sigma_t$ for minimizing~\eqref{eq:opt} is fixed over time. In other words, if $\sigma_{t_c}(i) = j$ for a fixed $i$ and $j$ at time $t_c$, then for any $t > t_{c}$ we also have $\sigma_{t_c}(i) = j$.
\end{proposition}
\begin{proof}
  We consider the case of two agents first, and then extend our approach to multiple agents through induction. We start with the case of two agents $x_1$ and $x_2$ and two targets $y_1$ and $y_2$. We will compare two policies: a time varying policy that includes at least one switch  $\hat{\mu}^{(2)}(t) = (\sigma_t, \hat{\mu}_1, \hat{\mu}_2),$ and a second policy $\mu^{(2)} = (\sigma, \mu_1, \mu_2)$ that does not switch and where $\mu_1$ and $\mu_2$ minimize Equation~\eqref{eq:1v1cost}. 

We first consider the case where $\hat{\mu}^{(2)}(t)$ both contains switches and the final assignment is equivalent to the initial assignment.  In this case let $t_1$ denote the time the first switch occurs and $t_2$ denote the time of the final switch back to the original assignment. Without loss of generality, assume that the initial assignment is given by $\sigma_0(i) = i$. Let $\hat{\mathcal{A}}(t)$ denote the active agents under the policy and $\hat{\tau}$ and $\hat{\tau}_i$ denote the exit times. The cost associated with $\hat{\mu}^{(2)}$ is
\begin{align*}
&J_{\hat{\mu}^{(2)}}(S(0)) = \int_{0}^{\hat{\tau}} \sum_{i=1,i \in \mathcal{\hat{A}}(t)}^n g_{i}(x_{i}, u_{i}, y_{\sigma_t(i)}) dt
\end{align*}
Letting $\delta_{i \in \mathcal{A}(t)}$ denote an indicator function, we can rewrite the total cost as a sum over each agent
\begin{align}
J_{\hat{\mu}^{(2)}}(S(0)) = \sum_{i=1}^{2} \int_{0}^{\hat{\tau}} \delta_{i \in \mathcal{A}(t)} g_{i}(x_i, \hat{\mu}_i(x_i), y_{\sigma_{t}(i)}) dt \label{eq:sep-delta} 
\end{align}
Now we can break up the integral into three section corresponding to the cost before the switch at $t_1$, between $t_1$ and $t_2$, and after $t_2$. Denoting $\hat{\tau_i}$ denote the exit time of agent $i$ we have 
\begin{align*}
J_{\hat{\mu}^{(2)}}(S(0)) &= \sum_{i=1}^{2} \Big[ \int_{0}^{t_1} g_{i}(x_i, \hat{\mu}_i(x_i), y_{\sigma(i)})dt + \\ 
& \qquad \qquad \int_{t_1}^{t_2} g_{i}(x_i, \hat{\mu}_i(x_i), y_{\sigma_t(i)})dt + \\
& \qquad \qquad \qquad \int_{t_2}^{\hat{\tau}_1} g_{i}(x_i, \hat{\mu}_i(x_i), y_{\sigma(i)} )dt \Big]
\end{align*}
Finally, suppose that $\mu^{(2)} = \hat{\mu}^{(2)}(0) = (\sigma_0, \mu_1, \mu_2)$ is a policy that maintained the original assignment. In this case, because $\mu_{1}$ and $\mu_2$ are optimized for the original assignment, they clearly result in policies that have lower costs than $\hat{\mu}_1$ and $\hat{\mu}_2$ incurred during the time period from $t_1$ and $t_2.$ In other words, since each agent ends up targeting the same target that it initially targeted, it is at least more effective to directly follow the policy to the target than to have intermediate deviations to the other target, i.e., 
\begin{equation*}
\resizebox{0.48\textwidth}{!}{$
    J_{\hat{\mu}^{(2)}}(S(0)) \geq \sum_{i=1}^{n} \int_{0}^{\tau_1} g_{i}(x_i, \mu_i(x_i), y_{\sigma(i)})dt = J_{\mu^{(2)}}(S(0)).
    $}
\end{equation*}
An identical argument follows for the case where the final policy is different than the initial policy. In that case, we would set $\mu^{(2)} = \hat{\mu}^{(2)}(\hat{\tau}).$

The case for more than two agents and targets follows by noticing that any system of $n$ agents can be analyzed by considering a system of two modified agents. The first modified agent is the augmentation of the first $n-1$, and the last one is the $n$th agent. 
Then the scenario is identical to the 2v2 assignment and the same argument follows.
\end{proof}

Now that we have shown that the optimal assignment is time independent, we can show that the minimizer of our our stated optimization problem~\eqref{eq:opt} is the same as that of the Monge problem~\eqref{eq:monge}.

\begin{theorem}[Optimization problem equivalence]
The optimal solution of the assignment problem given by the optimization problem~\eqref{eq:opt} is equivalent to that obtained by minimizing the Monge problem~\eqref{eq:monge} when the 1v1 cost function~\eqref{eq:cdyn} is used as the transportation cost.
\end{theorem}
\begin{proof}
 We use the fact that the optimal policy maintains a fixed index assignment vector for all time, i.e., $\sigma_t = \sigma$. Let $S(0)$ denote the initial state of the system, then the cost for any initial state can be represented as
\begin{align*}
J(S(0)) &= \int_{0}^{\tau} \sum_{i=1,i \in \mathcal{A}(t)}^n g_{i}(x_{i}, u_{i}, y_{\sigma(i)}) dt \\
        &= \sum_{i=1}^n \int_{0}^{\tau} \delta_{i \in \mathcal{A}(t)} g_{i}(x_{i}, u_{i}, y_{\sigma(i)}) dt  \\
        &= \sum_{i=1}^n \int_{0}^{\tau_{i}} g_{i}(x_{i}, u_{i}, y_{\sigma(i)}) dt \\ 
        &\geq \sum_{i=1}^n c_{\textrm{dyn}}(x_i(0), y_{\sigma(i)}(0))
\end{align*}
where the first equality came from Equations~\eqref{eq:opt_val} and~\eqref{eq:stage-sep} the second equality follows the same argument as Equation~\eqref{eq:sep-delta}; the third equality follows from the definition of $\tau_i$; and the final inequality follows from the definition of $c_{\textrm{dyn}}$ in Equation~\eqref{eq:cdyn}. Because of the definitions of 1v1 exit times $\tau_i$, we implicitly the cost function to only those policies where the agents reach their targets, i.e., $T_{\sharp}\mu_A = \mu_T$ where $\mu_A$ is the initial distribution of the agents and $\mu_T$ is the distribution of the targets at interception. 
Strict equality is obtained when the policies $\mu_i$ correspond to the optimal policies 1v1 policies that minimize~\eqref{eq:1v1cost} so that $u_i(t)$ are generated by $\mu_i$. Thus, we have proved the stated result.
\end{proof}

\section{Simulation Results}\label{sec:experiments}

We now numerically demonstrate the effectiveness of our approach through several simulated examples. In each example, we have used the 
Python Optimal Transport library~\cite{Flamary2017pot} to solve the underlying DOT problem. In each case, the dynamics are integrated via the RK45 integration scheme.

\subsection{Double integrators in three dimensions}

In this section we demonstrate that using the dynamics-based cost function over the standard distance-based Wasserstein metric yields significant savings that increases with size of the system. For the various examples we will consider agent/target systems of sizes  5 vs. 5, 10 vs. 10, 20 vs. 20 and 100 vs. 100. 

This set of examples uses a simple system of double integrators in three dimensions, where the velocity term is directly forced. The evolution of the state of each agent $x_i = (x_{i1}, x_{i2}, x_{i3}, x_{i4},  x_{i5}, x_{i6})$ is given by
\begin{equation*}
    \dot{x}_{ij} = x_{i(j+3)} 
\end{equation*}
for $j = 1,2,3$ and 
\begin{equation*}
    \dot{x}_{ij} = u_{i(j-3)}
\end{equation*} 
for $j = 4, 5, 6,$ where each agent has three control inputs (one for each dimension). The target dynamics are identical to the agent dynamics.

Each agent uses an infinite horizon linear-quadratic tracking policy of~\cite{willems2004rigorous} where the stage cost of an assignment is given by
\begin{equation}\label{eq:value_function_dyn}
\begin{split}
g(x_i, u_i, y_{\sigma(i)} = e_{i\sigma(i)}(t)^TQ_ie_{i\sigma(i)}(t) + u_i(t)^TR_iu_i(t) -\\ (e_{{ss}_{i\sigma(i)}}^TQ_ie_{{ss}_{i\sigma(i)}} + u_{{ss}_i}^TR_iu_{{ss}_i})
\end{split} 
\end{equation}
where $e_{i\sigma(i)}(t)$ and $e_{{ss}_{i\sigma(i)}}$ are defined as the transient error and steady-state error between the $ith$ agent state and an assigned-to target state, respectively; $u_i$ is the control input to drive the $ith$ agent to the assigned-to target; and $u_{{ss}_i}$ is the control input for the $ith$ agent to the assigned-to target operating at steady-state conditions. For the weight matrices we choose $Q_i = \textrm{diag}(10^3, 10^3, 10^3, 0, 0, 0),$ where the nonzero weights correspond to the errors in positions in each dimension and the zero weights correspond to the errors in velocity. The control penalty is chosen to be $R_i = I_{3 \times 3}.$ The targets use an identical tracking policy; however they track certain fixed positions in space. 

The initial conditions of the system consist of the positions and velocities of each agent and target, a set of stationary locations that are tracked by the targets, and a set of assignments from each target to the the stationary location. These conditions are randomly generated for the following results.

The initial conditions of the agents consist of uniformly distributed position and velocity components on an interval of $-1000$ to $1000$ and $-5000$ to $5000$, respectively. The initial conditions of the targets position are equivalent, but with velocity parameters following a uniform distribution between $-1000$ to $1000$. The terminal target locations were randomly selected on a uniform distribution between $-1000$ and $1000$. 

\begin{figure}[h!]
\centering
\includegraphics[width=0.5\textwidth]{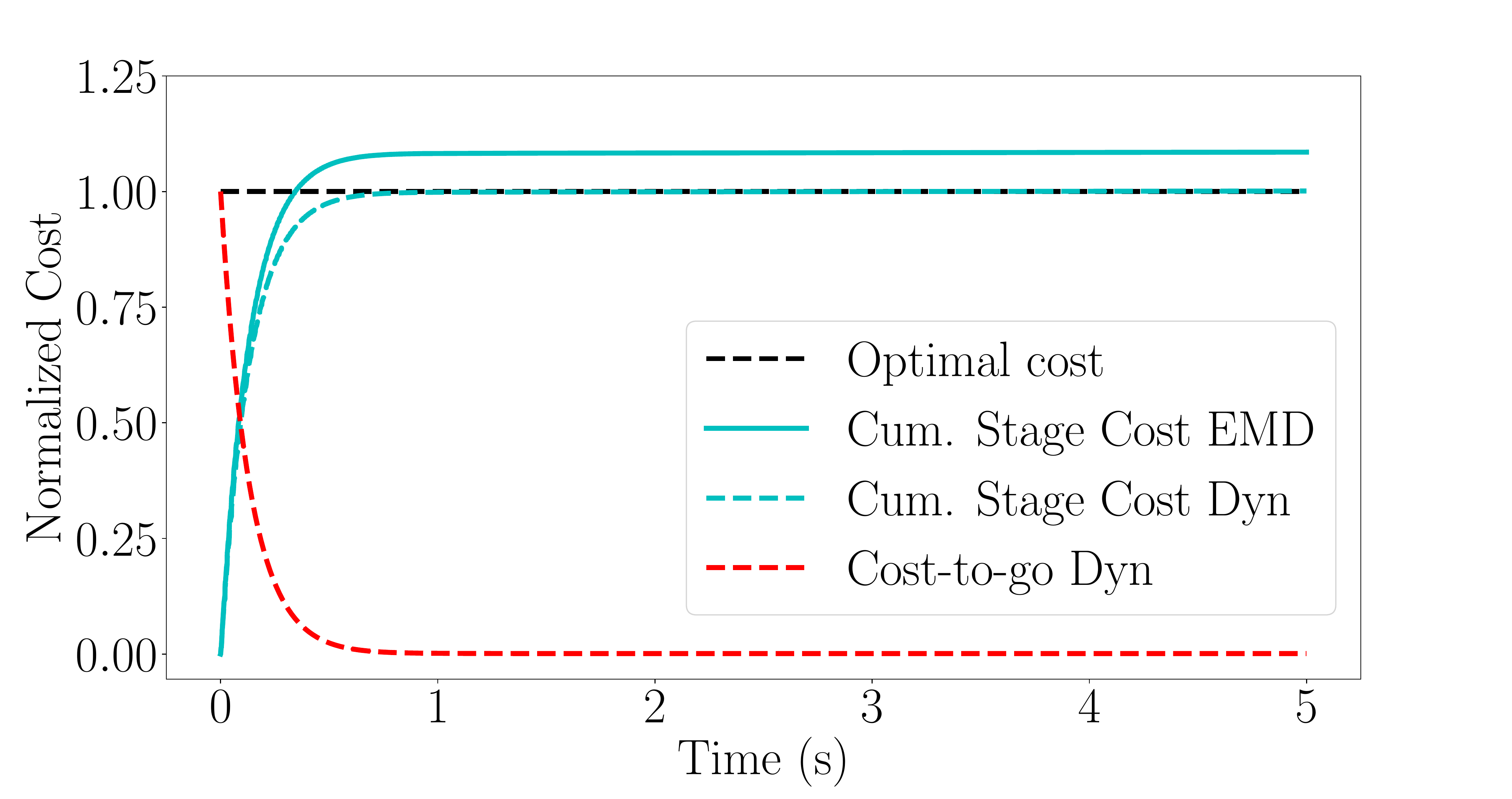}
\caption{Normalized costs incurred by 100 agents tracking 100 targets. The cumulative costs over time for the assignment policy that uses EMD is shown by the solid cyan line. The cost of this policy exceeds the optimal cost given by the dotted black line. This optimal cost is computed by summing the optimal costs of the value function for each agent under the optimal assignment. The cumulative costs of the dynamics-based assignment approaches this optimal value, as expected.}
\label{fig:cumcost}
\end{figure}

In Figure~\ref{fig:cumcost}, we show the cumulative control costs incurred by a system of 100 agents while they are attempting to tracking 100 targets. Recall that the EMD-based objective assigns agents to targets with the aim of minimizing the total Euclidean distance. This assignment does not account for the dynamics of the agent and as a result, it performs worse than the dynamics-based assignment which accounts for the effort to actually get the agent to its assigned target. Mechanically, this performance difference results because agents are either incorrectly assigned at the beginning or because agents switch assignments over the course of their operations. For this simulation, the EMD-based policy checks whether reassignment is necessary every 0.1 seconds. 

Because visualizing the movements of 100 agents and targets is difficult, we demonstrate prototypical movements for a 5 vs. 5 system in Figures~\ref{fig:trajectory} and its X-Y projection ~\ref{fig:2dtrajectory} These figures demonstrate both the optimal trajectories of the agents and targets under the dynamics-based optimal assignment and the sub-optimal trajectories of the EMD-based assignment. Agents A0 and A1, for example, take significantly different paths to different targets. The movements corresponding to the EMD-based policy require more manuevering.  

The dynamics-based policy leverages the dynamic capabilities of the agents to select the targets that each individual would optimally be able to track over time. Finally, note that the individual agent controllers that we use are fundamentally \textit{tracking} controllers, thus the agents act to match the velocity and position of their targets. This is the reason why several maneuvers show the agent passing and then returning to the target --- for instance A1 to T3 under the EMD policy.

\begin{figure}[h!]
\centering
\includegraphics[width=0.5\textwidth]{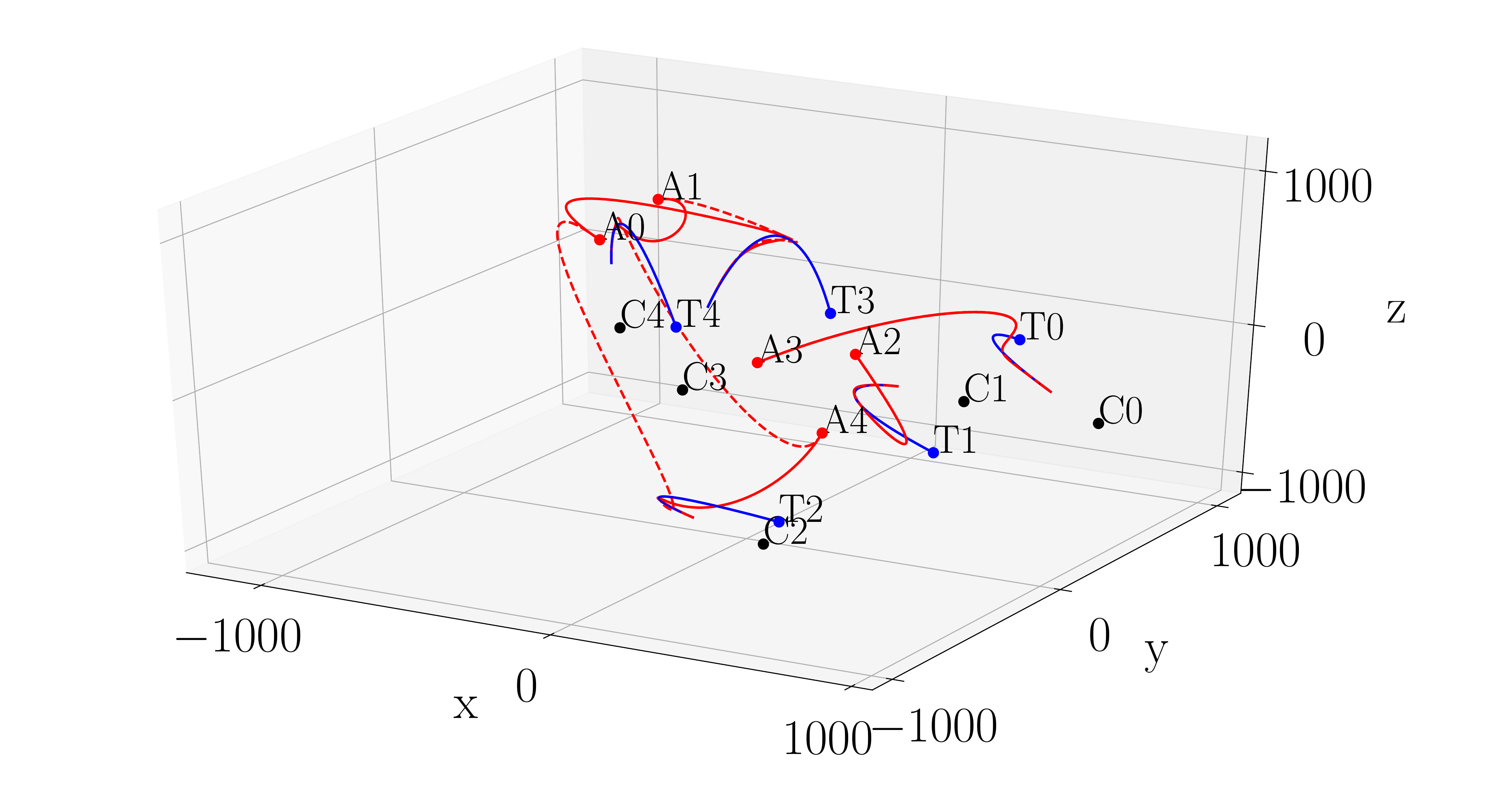}
\caption{Trajectories of agents and targets in a 5 vs 5 system. The trajectories for an assignment policy that accounts for the dynamics (Dyn) are qualitatively different than the assignment policy that uses the Wasserstein distance (EMD). For labels of each path see the X-Y projection in Figure~\ref{fig:2dtrajectory}.}
\label{fig:trajectory}
\end{figure}

\begin{figure}[h!]
\centering
\includegraphics[width=0.5\textwidth]{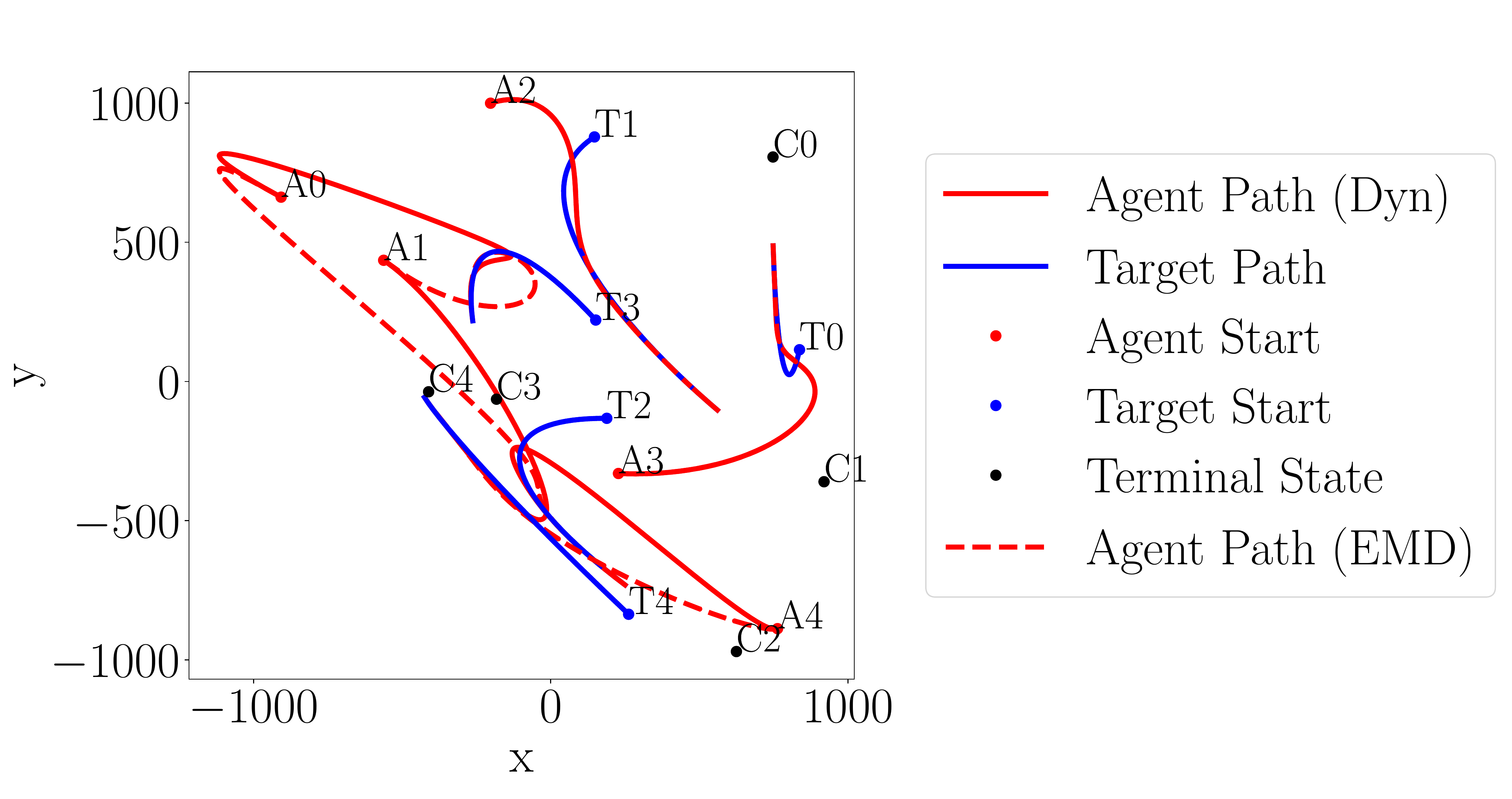}
\caption{Projection of the trajectories of Figure~\ref{fig:trajectory} onto the X-Y plane.}
\label{fig:2dtrajectory}
\end{figure} 
Since the dynamics-based assignment policy selects an optimal assignment at initial time, it offers significant control cost benefits over assignments that continually reassign the agents based on the EMD.

The benefit of dynamics-based assignment grows with the size of the system. To demonstrate this fact we use perform Monte Carlo simulations of one hundred realizations of a 5 vs 5, 10 vs 10, and 20 vs 20 system by sampling over initial conditions. As the complexity of the engagement increases, the amount of additional control effort required by the EMD based assignment grows, shown by  Figure~\ref{fig:cost_histogram}.

Furthermore,  Figure~\ref{fig:assignment_avg_switches} illustrates that as the system size grows the EMD-based policy performs more switches. This fact contributes to the observed loss in efficiency of the EMD-based policy.

\begin{figure}[ht!]
\centering
\includegraphics[width=0.5\textwidth]{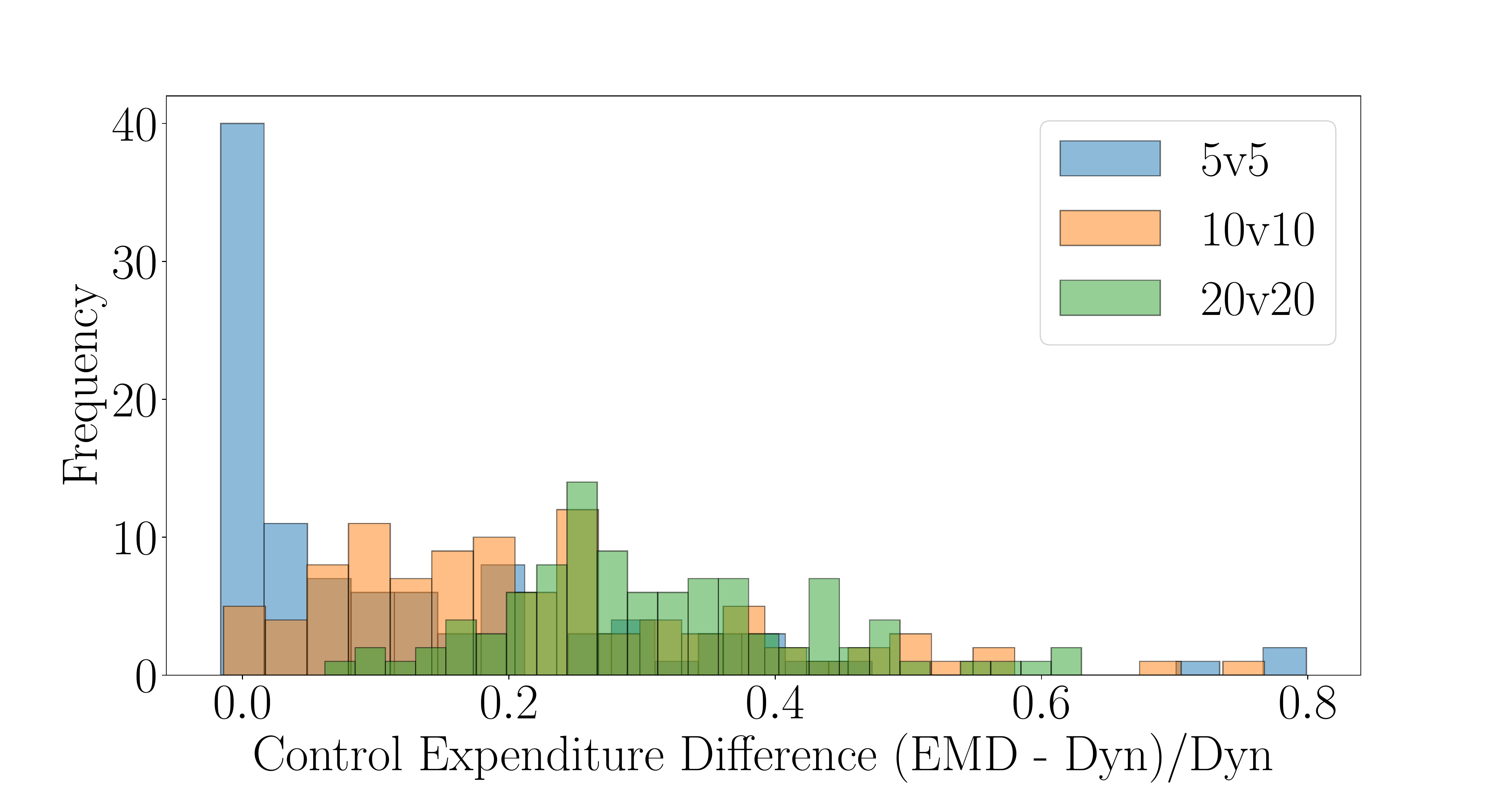}
\caption{The EMD-based assignment policy becomes increasingly less effective as the size of the system grows. Histograms of 100 Monte Carlo simulations obtained by sampling initial conditions for various system sizes are shown. As the system size increases, the distribution of the difference between the control effort of the EMD and Dyn based assignments increases.}
\label{fig:cost_histogram}
\end{figure}

\begin{figure}[ht!]
\centering
\includegraphics[width=0.5\textwidth]{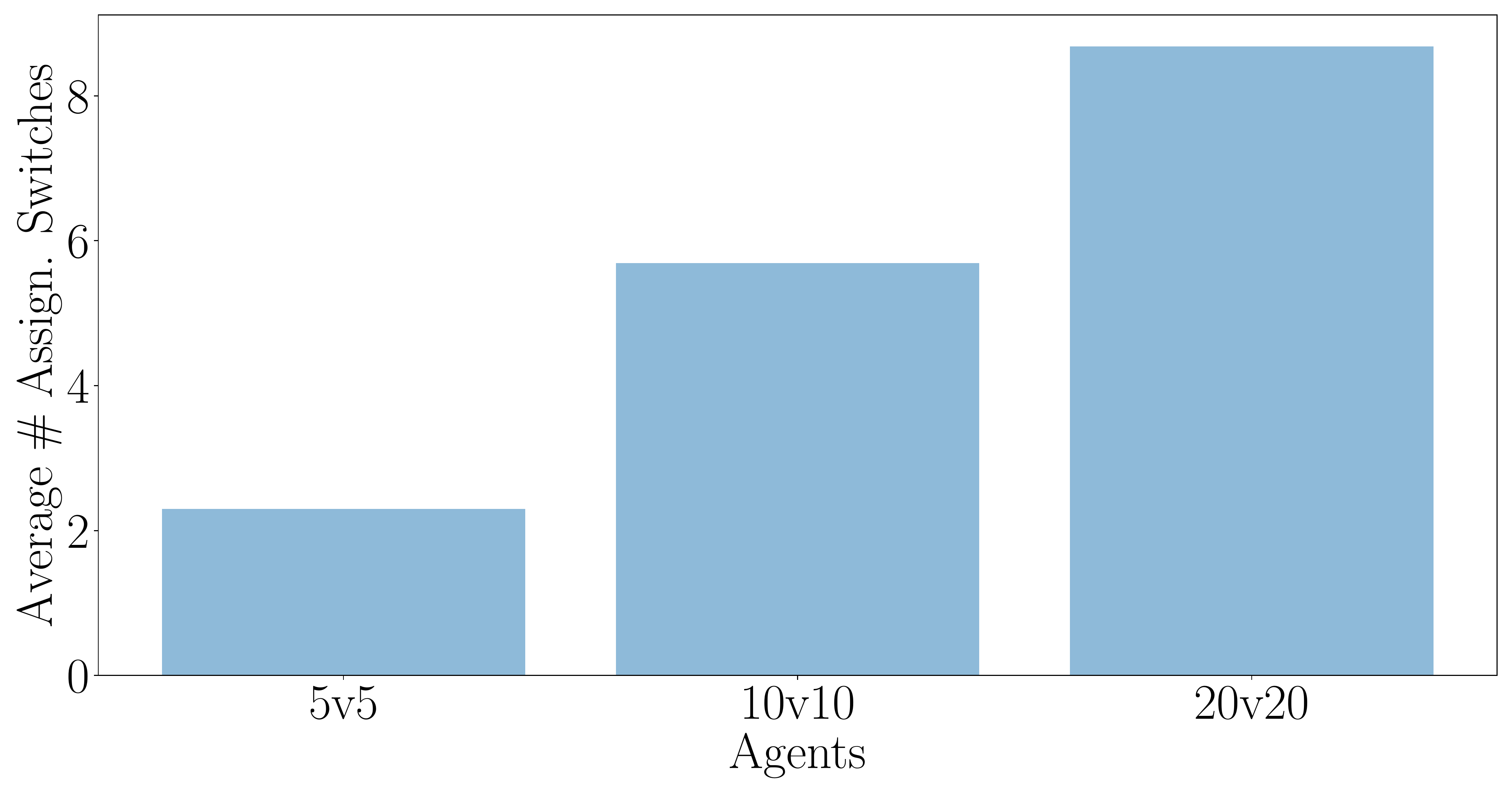}
\caption{ Monte Carlo simulations for 5v5, 10v10, and 20v20 systems reveal that the average number of Agent-Target assignment switches for EMD-based assignments positively correlates with the size of the systems. }
\label{fig:assignment_avg_switches}
\end{figure} 

\subsection{Linearized Quadcopter}
We now compare the algorithms on swarms of linearized quadcopter dynamics \cite{sabatino2015quadrotor} that are slightly modified versions of double integrators. The dynamics of both the agents and the targets in this case are given by

\begin{align*}
    \mathbf{\hat{f}(x,u,d)} =
    \begin{cases}
           \dot{x} = u \\
           \dot{y} = v \\
           \dot{z} = w \\
           \dot{\psi} = r \\
           \dot{\theta} = q \\
           \dot{\phi} = p \\
           \dot{u} = -g\theta + \frac{f_{wx}}{m} \\
           \dot{v} = g\phi + \frac{f_{wy}}{m} \\
           \dot{w} = \frac{f_{wx} - f_t}{m} \\
           \dot{p} = \frac{\tau_x + \tau_{wx}}{I_{xx}} \\
           \dot{q} = \frac{\tau_y + \tau_{wy}}{I_{yy}} \\ 
           \dot{r} = \frac{\tau_z + \tau_{wz}}{I_{zz}}
    \end{cases}
\end{align*}
where the twelve dimensional state space
\begin{align*}
    \mathbf{x} = 
    \begin{bmatrix}
        x & y & z & \psi & \theta & \phi & u & v & w & p & q & r
    \end{bmatrix}^T
    \in \mathbf{R}^{12}
\end{align*}

 consists of the position, attitude, translational velocity, and rotational velocity components of the vehicle. The parameters of the system are $m = 0.1 \textrm{kg}$, $I_{xx} = 0.00062 \textrm{kg-m}^2$, $I_{yy} =  0.00113 \textrm{kg-m}^2$, $I_{zz} = 0.9(I_{xx} + I_{yy})$ and $g = 9.81 \textrm{m/s}^2,$ respectively. Linearization was performed under small oscillation and small angle approximations. Furthermore, we will assume no wind disturbance forces and torques, $d = \left[f_{wx}, f_{wy}, f_{wz}, \tau_{wx}, \tau_{wy}, \tau_{wz}\right] = 0$. The control inputs are four dimensional and consist of the forces and torques $u = \left[f_t, \tau_x, \tau_y, \tau_z\right]$ that act on the vertical thrust and angular motions about the three principal axes.

 The initial positions and velocities of the agents are sampled uniformly between $-100$ to $100$ and $-500$ to $500$, respectively. The initial velocities of the targets were sampled uniformly between $-50$ to $50$. The attitude and rotational velocity terms for both agents and targets were uniformly distributed between $-2\pi$ and $2\pi$ and $-25$ and $25$, respectively. The terminal target locations were randomly selected from a uniform distribution between $-100$ and $100$. The control parameters for the agents and targets are updated to $Q_i = diag(10^3, 10^3, 10^3, 10^3, 10^3, 10^3, 0, 0, 0, 0, 0, 0)$ and $R_i = I_{4x4}$, respectively.

Similar to the double integrator systems, the dynamics-based assignment policy is able to optimally assign the more complex quadcopter agents to complete their tracking task with minimal cost. Figure~\ref{fig:cumcost_quadcopter} illustrates the cumulative cost expended by the agent swarm and once again showcases the optimality of the dynamics-based assignment method. Unlike the EMD policy, the complete dynamic information of the swarm members are used in the decision-making process as opposed to only the euclidean distance components. In the end, the EMD-based assignment policy incurs a cost that is 1.7 times greater than the dynamics-based assignment policy.

Figures~\ref{fig:trajectory_quadcopter} and Figures~\ref{fig:2dtrajectory_quadcopter} reveal the paths taken by the agents managed by the EMD and Dyn policies. Agents 1 and 4, in particular, are allowed to take advantage of their initial dynamic states to cheaply track their targets, instead of being reassigned (by the EMD-based policy) mid-flight to closer targets that appear. In this case, the reassignment causes extreme turning maneuvers that require significant control expense. %

\begin{figure}[ht!]
\centering
\includegraphics[width=0.5\textwidth]{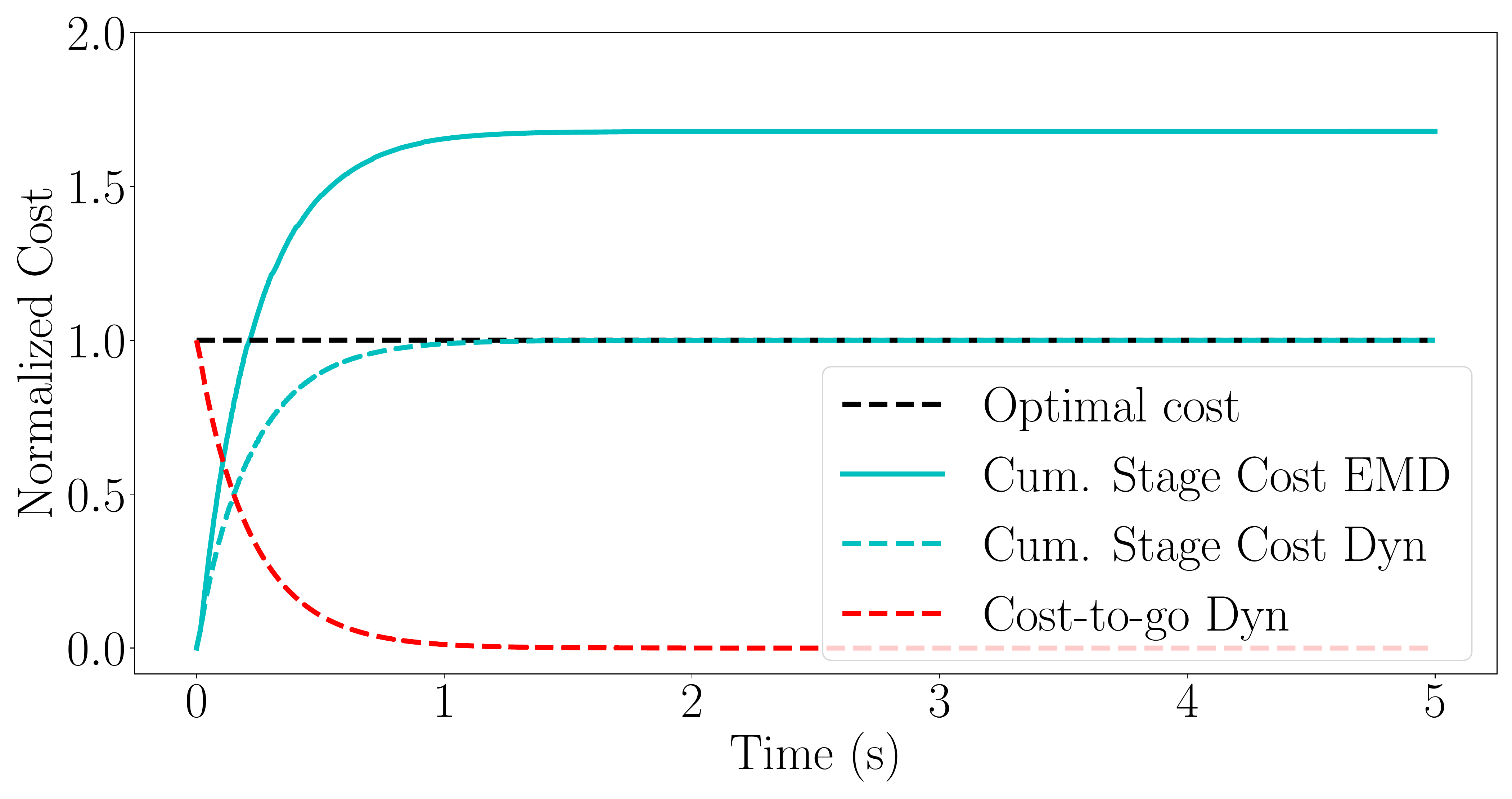}
\caption{Normalized costs incurred by 5 linearized quadcopter agents tracking 5 linearized quadcopter targets. The cumulative costs for the EMD policy exceeds the optimal cumulative costs of the dynamics-based policy for a system operating realistic dynamics. The dynamics-based policy continues to settle at the optimal value.}
\label{fig:cumcost_quadcopter}
\end{figure}

\begin{figure}[ht!]
\centering
\includegraphics[width=0.5\textwidth]{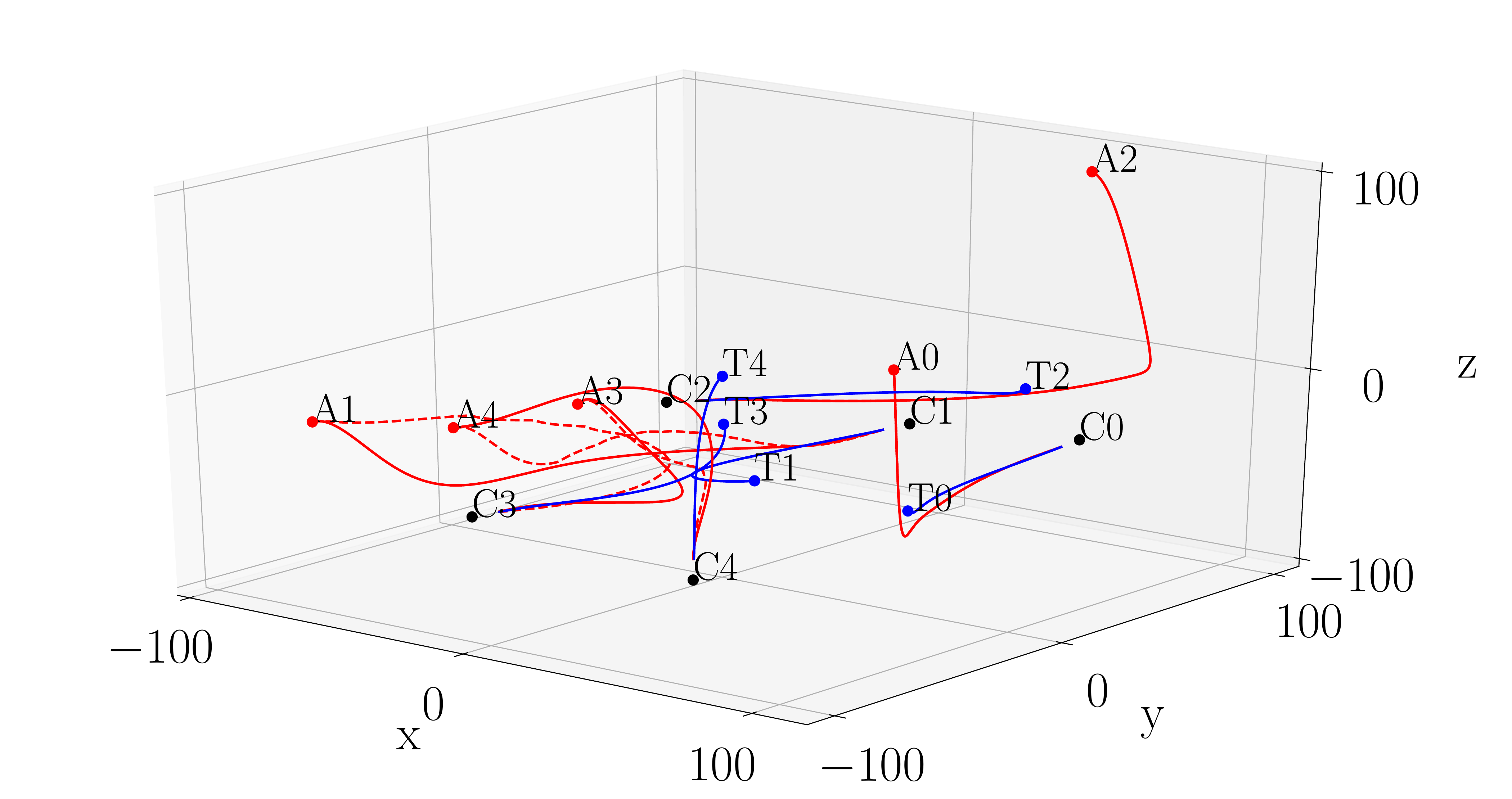}
\caption{Trajectories for a 5 vs 5 system operating linearized quadcopter dynamics. The dynamics-based policy accounts for the full dynamic capability of the agents in its assignment, rather than solely relying on spatial proximity information. This includes leveraging the rotational and translational information of the vehicle in the decision process.}
\label{fig:trajectory_quadcopter}
\end{figure}

\begin{figure}[h!]
\centering
\includegraphics[width=0.5\textwidth]{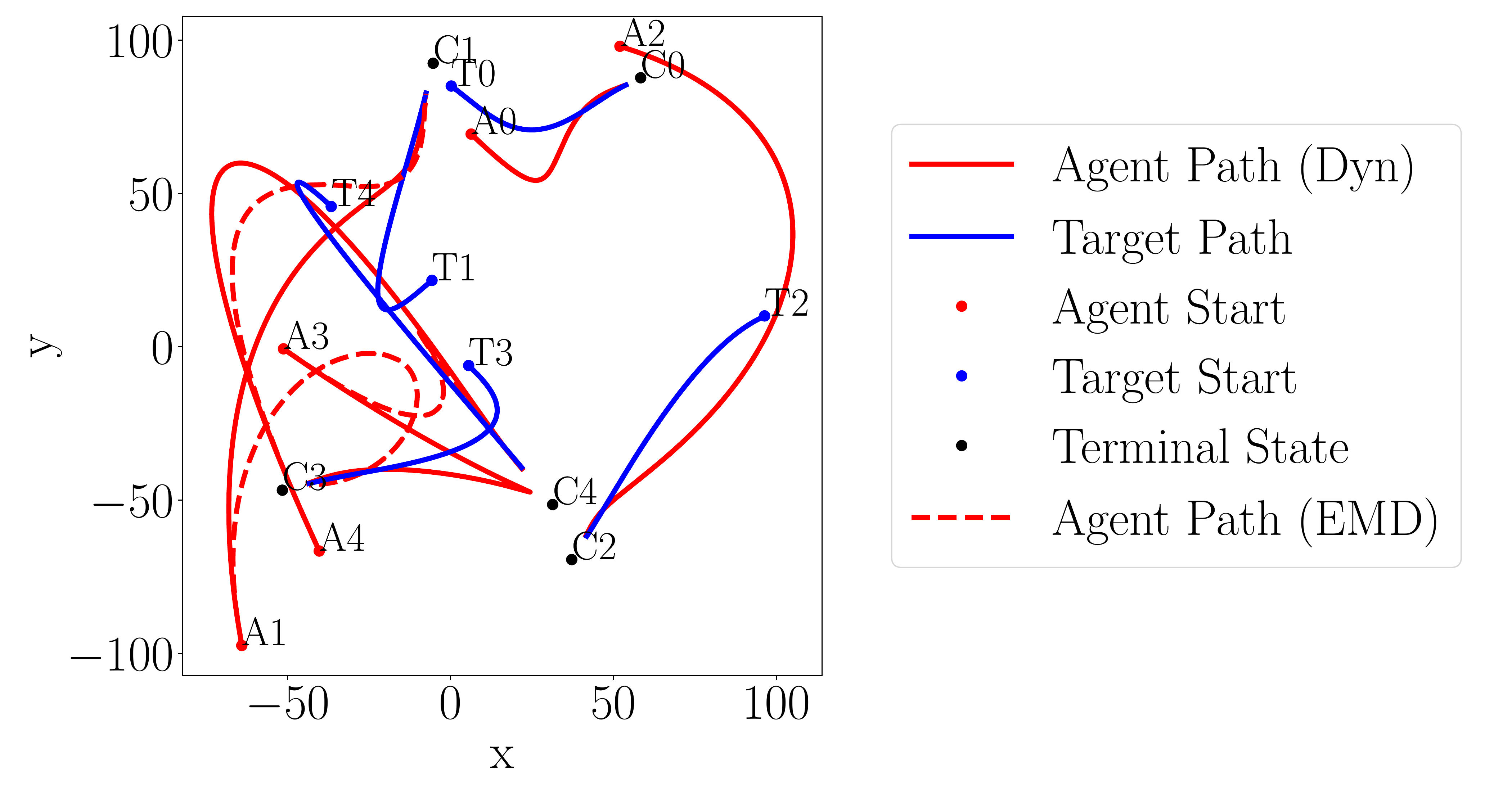}
\caption{Projection of the trajectories of Figure~\ref{fig:trajectory_quadcopter} onto the X-Y plane. Agents 1 and 4 are able to use their initial dynamic conditions to optimally track their targets instead of performing expensive turning maneuvers.}
\label{fig:2dtrajectory_quadcopter}
\end{figure} 

Since the linearized quadcopter operates over a $d=12$ statespace, the computational cost for performing assignments are more expensive, and since the EMD-based policy requires checking and updating assignments every time increment, it requires significantly greater computational expense. For this problem, the total cost of all reassignments required 0.6 seconds by the EMD policy, a signification porition of the total simulation time of five seconds.

\section{Conclusion}
In this paper we have demonstrated how to reformulate a dynamic multi-vehicle assignment problem into a linear program by linking this problem with the theory of optimal transport. This theory allows us to prove optimality and to increase the system efficiency using our approach. In the end, we have developed an assignment approach that is capability-aware. The assignment accounts for the capabilities of all the agents and targets in the system.

One direction of future research is the incorporation of constraints amongst the various agents to avoid collisions or other interactions. An extension of DOT theory in this direction could greatly increase the tractability of numerous multi-agent swarm operations, for example large scale formation flight. Another direction for future research is the incorporation of stochastic dynamics and partial state information. For either case, the approach described in this paper can be used as the basis of a greedy or approximate dynamic programming approach that is traditionally used for these problems. Finally, we can incorporate learning into the program where the agents periodically update their knowledge about the intent of the targets.

\section{ACKNOWLEDGMENTS}
We would like to thank Tom Bucklaew and Dustin Martin of Draper Laboratory for their helpful guidance and vision in support of this project. This research has been supported by Draper Laboratory, 555 Technology Square, Cambridge, MA 02139.


\bibliographystyle{plain}
\bibliography{main}

\end{document}